\documentclass[11pt]{article}

\usepackage{a4}

\usepackage[T1]{fontenc}
\usepackage{flafter}

\usepackage{enumitem}

\usepackage{xspace}
\usepackage{float}
\floatstyle{boxed}
\usepackage{framed}
\usepackage{amsmath,amssymb,amsthm}

\newcommand{\R}{\mathbb{R}}

\newcommand{\dotp}[2]{\langle #1 , #2 \rangle}

\newtheorem{claim}{Claim}
\newtheorem{theorem}{Theorem}
\newtheorem{lemma}{Lemma}

\newtheorem{fact}{Fact}
\newtheorem{proposition}{Proposition}
\newtheorem{definition}{Definition}
\newtheorem{problem}{Problem}
\newtheorem{remark}{Remark}

\newcommand{\ggi}{\textsc{GGI}\xspace}

\newcommand{\gaut}{\textsc{Geom{-}AUT}\xspace}

\newcommand{\evgi}{\textsc{EVGI}\xspace}

\newcommand{\haut}{\textsc{Hyp{-}AUT}\xspace}

\newcommand{\Q}{\mathbb{Q}}

\renewcommand{\P}{\mathcal{P}}

\newcommand{\proj}{\textrm{proj}}

\newcommand{\sym}{\textrm{Sym}}

\newcommand{\Aut}{\textrm{Aut}}

\newcommand{\NP}{\textsf{NP}}

\newcommand{\coAM}{\textsf{coAM}}

\newcommand{\ISO}{\textrm{ISO}}

\newcommand{\cmp}{O^*(k^{O(k)})\xspace}

\title{Faster FPT Algorithm for Graph Isomorphism Parameterized by
  Eigenvalue Multiplicity}

\author{V. Arvind and Gaurav Rattan\\
        The Institute of Mathematical Sciences\\
        C.I.T. Campus\\
        Chennai 600 113, India\\
        \{arvind,grattan\}@imsc.res.in\\}\date{}

\date{}

\begin{document}

\maketitle

\thispagestyle{plain}
\pagestyle{plain}

\begin{abstract}
We give a $\cmp$ time isomorphism testing algorithm for graphs of
eigenvalue multiplicity bounded by $k$ which improves on the previous
best running time bound of $O^*(2^{O(k^2/\log
  k)})$~\cite{ep}.\footnote{Throughout the paper, we use the
  $O^*(\dot)$ notation to suppress multiplicative factors that are
  polynomial in input size.}
\end{abstract}

\section{Introduction}\label{one}

Two simple undirected graphs $X=(V,E)$ and $X'=(V',E')$ are said to be
\emph{isomorphic} if there is a bijection $\varphi:V\rightarrow V'$
such that for all pairs $\{u,v\}\in{V\choose 2}$, $\{u,v\}\in E$ if
    and only if $\{\varphi(u),\varphi(v)\}\in E'$. Given two graphs
    $X$ and $X'$ as input the decision problem \emph{Graph
      Isomorphism} asks whether $X$ is isomorphic to $X'$. An
    outstanding open problem in the field of algorithms and complexity
    is whether the Graph Isomorphism problem has a polynomial-time
    algorithm.  The asymptotically fastest known algorithm for Graph
    Isomorphism has worst-case running time time $2^{O(\sqrt{n\lg
        n})}$ on $n$-vertex graphs \cite{BL83}. On the other hand, the
    problem is unlikely to be NP-complete as it is in $\NP\cap\coAM$
    \cite{BHZ}.

However, efficient algorithms for Graph Isomorphism have been
discovered over the years for various interesting subclasses of
graphs, like, for example, bounded degree graphs \cite{Luks}, bounded
genus graphs \cite{Miller,Grohe}, bounded eigenvalue multiplicity
graphs \cite{bgm,ep}. 

The focus of the present paper is Graph Isomorphism for bounded
eigenvalue multiplicity graphs. This was first studied by Babai et al
\cite{bgm} who gave an $n^{O(k)}$ time algorithm for it. There is also
an NC algorithm\footnote{NC denotes the class of problems that can be
  solved in in the parallel-RAM model in polylogarithmic time using
  polynomially many processors.} for the problem for constant $k$ due
to Babai \cite{bab}. Using an approach based on cellular algebras and
some nontrivial group theory, Evdokimov and Ponomarenko \cite{ep} gave
an $O^*(2^{O(k^2/\log k)})$ algorithm for it. This puts the problem in
FPT, which is the class of \emph{fixed parameter tractable}
problems. The parameter in question here is the bound $k$ on the
eigenvalue multiplicity of the input graphs.

In this paper we obtain a $\cmp$ time isomorphism algorithm for graphs
of eigenvalue multiplicity bounded by $k$. We follow a relatively
simple geometric approach to the problem using integer
lattices. Recently, we obtained an $\cmp$ time algorithm for
\emph{Point Set Congruence} (abbreviated $\ggi$) in $\Q^k$ in the
$\ell_2$ metric \cite{ar}. Our algorithm is based on a lattice
isomorphism algorithm of running time $\cmp$, due to Haviv and Regev
\cite{HR13}. They design an $O^*(n^{O(n)})$ time algorithm for
checking if two integer lattices in $\R^n$ are isomorphic under an
orthogonal transformation. In \cite{ar} we adapt their technique to
solve the Point Set Congruence problem, $\ggi$, in $\cmp$ time.

Now, in this paper, building on our previous algorithm for $\ggi$
\cite{ar}, combined with some permutation group algorithms, we first
give an $O^*(k^{O(k)})$ time algorithm for a suitable \emph{geometric
  automorphism} problem, defined in Section~\ref{four}. It turns out
that the bounded eigenvalue multiplicity Graph Isomorphism can be
efficiently reduced to this geometric automorphism problem, which
yields the $\cmp$ time algorithm for it.


\section{Preliminaries}\label{two}

Let $[n]$ denote the set $\{1,\dots,n\}$.  We assume basic familiarity
with the notions of vector spaces, linear transformations and
matrices.  The projection of a vector $v\in\R^n$ on a subspace
$S\subseteq \R^n$ is denoted as $proj_S(v)$. The \emph{inner product}
of vectors $u=(u_1,\dots,u_n)$ and $v=(v_1,\dots,v_n)$ is $\dotp{u}{v}
= \displaystyle\sum_{ i \in [n]} u_i v_i $. The \emph{euclidean norm},
$\|u\|$, of a vector $u$, is $\sqrt{\dotp{u}{u}}$, and the
\emph{distance} between two points $u$ and $v$ in $\R^n$ is $\|u -
v\|$. Vectors $u,v$ are \emph{orthogonal} if $\dotp{u}{v} =
0$. Subspaces $U,V$ are orthogonal if for every $u \in U, v\in V$,
$u,v$ are orthogonal.  A set of subspaces $W_1,\dots,W_r$ is said to
be an \emph{orthogonal decomposition} of $\R^n$ if each pair of
subspaces are mutually orthogonal, and they span $\R^n$.  A square
matrix $M$ is orthogonal if $M^TM = I$. A linear transformation $T$
\emph{stabilizes} a subspace $S$ if $T(S) \subseteq S$.  Given a
matrix $M$, we call $\lambda$ to be an \emph{eigenvalue} of $M$ if
there exists a vector $v$ such that $Mv=\lambda v$.  We call $v$ to be
an \emph{eigenvector} of $M$ of eigenvalue $\lambda$.  The set of all
eigenvectors of $M$ of eigenvalue $\lambda$ is a subspace of $\R^n$.
The following well-known fact about $n \times n$ symmetric matrices
will be useful.

\begin{fact}
All eigenvalues of a symmetric matrix are real. Moreover, the
eigenspaces form an orthogonal decomposition of $\R^n$.
\end{fact}


We use $\sym(V)$ to denote group of all permutations on a finite set
$V$. Given a graph $X=(V,E)$, a permutation $\pi \in \sym(V)$ is an
\emph{automorphism} of the graph $X$ if for all pairs $\{u,v\}$ of
  vertices, $\{u,v\} \in E$ iff $\{\pi(u),\pi(v)\} \in E$. In other
    words, $\pi$ preserves adjacency in $X$. The set of all
    automorphisms of a graph $X$, denoted by $\Aut(X)$, is a subgroup
    of $\sym(V)$, which is denoted by $\Aut(X)\le \sym(V)$.

We can similarly talk of automorphisms of hypergraphs: Let $X=(V,E)$
be a hypergraph with vertex set $V$ and edge set $E\subset 2^E$.  A
permutation $\pi \in \sym(V)$ is an \emph{automorphism} of the
\emph{hypergraph} $X$ if for every subset $e\subseteq V$, $e\in E$ if
and only if $\pi(e)\in E$, where $\pi(e)=\{\pi(v)\mid v\in e\}$.


Given an undirected graph $X=(V,E)$, the set $V$ indexed by $[n]$, we
define its \emph{adjacency matrix} $A_X$ is defined as follows:
$A_X(i,j) = 1$ if $\{v_i,v_j\} \in E$ and $0$ otherwise. Clearly, the
adjacency matrix $A_X$ of an undirected graph $X$ is symmetric. Given
a permutation $\pi:[n]\rightarrow[n]$, we can associate a natural
permutation matrix $M_{\pi}$ with it.  It is easy to verify that $\pi$
is an automorphism of a graph $G$ iff $M_{\pi}^T A_X M_{\pi} = A_X$.
Since permutation matrices are orthogonal matrices, the following
simple folklore lemma characterizes the automorphisms of a graph
through the action of the associated matrix on the eigenspaces of its
adjacency matrix.

\begin{lemma}\label{prelim1}
Let $X$ be the adjacency matrix of a graph $G=(V,E)$.
Then, a permutation $\pi \in \sym(V)$ is an automorphism of $G$
iff the associated linear map $M_{\pi}$ preserves the eigenspaces of $X$.
\end{lemma}

\begin{proof}
Suppose $\pi \in \Aut(G)$. Then $M_{\pi} A_X = A_X M_{\pi}$ and
therefore, for any eigenvector $v$ in eigenspace $W_i$ of eigenvalue
$\lambda_i$, $A_X M_{\pi} v = M_{\pi} A_X v = \lambda_i M_{\pi} v$
which shows that $M_{\pi}v \in W_i$.  Conversely, suppose $M_{\pi}$
preserves eigenspaces $W_i$ of $X$.  Then, for any $v \in W_i$, $A_X
M_{\pi} v = \lambda_i M_{\pi}x = M_{\pi} A_X v $.  Since eigenvectors
of the symmetric matrix $A_X$ span $\R^n$, this implies that
$A_XM_{\pi} = M_{\pi}A_X$.  Therefore, $\pi$ must be an automorphism
of $G$.  
\end{proof}

\begin{remark}
Our approach to solving Graph Isomorphism for bounded eigenvalue
multiplicity is based on a variation of this lemma, as described in
Proposition~\ref{prop1}. We first map the graph $G$ into a point set
$\P$ in the $n$-dimensional space $\R^n$. Then, we project $\P$ into
eigenspace $W_i$ of $G$, to obtain $\P_i$, for each eigenspace $W_i$.
It turns out that $\pi$ is an automorphism of $G$ if and only if
$\pi$, in its induced action is a congruence for the point set $\P_i$
for each eigenspace $W_i$. When the eigenspaces $W_i$ are of dimension
bounded by the parameter $k$, it creates the setting for application
of the $\cmp$-time algorithm for $\ggi$ \cite{ar}.
\end{remark}

Next, we recall some useful results about permutation group
algorithms. Further details can be found in the excellent text of
Ser\'ess \cite{Seress}. 

A \emph{permutation group} is a subgroup $G\le \sym(\Omega)$ of the
group of all permutations on a finite domain $\Omega$. A subset
$A\subseteq G$ of a permutation group $G$ is a \emph{generating set}
for $G$ if every element of $G$ can be expressed as a product of
elements of $A$. Every permutation group $G\le \sym(\Omega)$ has a
generating set of size $\log |G|\le n\log n)$ where
$n=|\Omega|$. Thus, algorithmically, a compact input representation
for permutation groups is by a generating set of size at most $n\log
n$. With this input representation, it turns out there several natural
permutation group problems have efficient polynomial-time algorithms.
A fundamental problem here is \emph{membership testing}: Given a
permutation $\pi\in\sym(\Omega)$ and permutation group $G$ by a
generating set, there is a polynomial-time algorithm (the
Schreier-Sims algorithm~\cite{Seress}) to check if in $\pi\in G$.  The
\emph{pointwise stabilizer} of a subset $\Delta \in \Omega$ in a
permutation group $G\le \sym(\Omega)$ is the subgroup
\[
G_{\{\Delta\}}=\{\pi \in G \mid \forall \gamma \in
\Gamma,~\pi(\gamma)=\gamma\}.
\]

Given a permutation group $G\le\sym(\Omega)$ by a generating set, a
generating set for $G_{\{\Delta\}}$ in polynomial time using ideas
from the Schreier-Sims algorithm~\cite{Seress}. More generally, suppose
$G\le\sym(\Omega)$ is given by a generating set and
$\sigma\in\sym(\Omega)$ is a permutation. The subset of permutations
$(G\sigma)_{\Delta\}}=\{\pi\in G\sigma\mid \pi(\gamma)=\gamma \forall
\gamma\in\Delta\}$ that pointwise fix $\Delta$ is a right coset
$G_{\{\pi^{-1}(\Delta\})}\tau$ and a generating set for
$G_{\{\pi^{-1}(\Delta\})}$ and such a coset representative $\tau$ can
be computed in polynomial time \cite{Seress}. We often use the
following group-theoretic fact.

\begin{fact}\label{union}
Let $H_i\le \sym(\Omega), 1\le i\le t$ and $\sigma_i\in\sym(\Omega),
1\le i\le t$, where each $H_i$ is given by a generating set $A_i$.
Suppose the union of the right cosets $\bigcup_{i=1}^t H_i\sigma_i$ is
a coset $G\sigma$ for some subgroup $G\le\sym(\Omega)$. Then, we can
choose the coset representative $\sigma$ to be $\sigma_1$ and the set
$\bigcup_{i=1}^t A_i\cup \{\sigma_i\sigma_1^{-1}\mid 2\le i\le t\}$ is
  a generating set for $G$.
\end{fact}


\section{Algorithm Overview}\label{three}

Before we give an overview of the main result of this paper, we recall
the Point Set Congruence problem (also known as the geometric
isomorphism problem) $\ggi$ \cite{A+88,Aku,BK}.

Given two finite $n$-point sets $A$ and $B$ in $\Q^k$, we say $A$ and
$B$ are \emph{isomorphic} if there is a \emph{distance-preserving}
bijection between $A$ and $B$, where the distance is in the $l_2$
metric. The \emph{Geometric Graph Isomorphism} problem, denoted
$\ggi$, is to decide if $A$ and $B$ are isomorphic. This problem is
also known as \emph{Point Set Congruence} in the computational
geometry literature \cite{Aku,BK,A+88}. It is called ``Geometric Graph
Isomorphism'' by Evdokimov and Ponomarenko in \cite{epggi}, which we find
more suitable as the problem is closely related to Graph Isomorphism.
In \cite{ar} we obtained a $\cmp$ time algorithm for this problem.

We now begin with a definition.

\begin{definition}
Let $\P=\{p_1,p_2,\ldots,p_m\}\subset \Q^n$ be a finite point set.
A \emph{geometric automorphism} of $\P$ is a permutation $\pi$ of the
point set $\P$ such that for each pair of points $p_i,p_j\in\P$ we
have
\begin{eqnarray*}
\|p_i\|& = & \|\pi(p_i)\|, \textrm{ and}\\
\|p_i-p_j\| & = &\|\pi(p_i)-\pi(p_j)\|,
\end{eqnarray*}
where $p_i$ denotes, by abuse of notation, also the position vector of
the point $p_i$.
\end{definition}

Let $\P=\{p_1,p_2,\ldots,p_m\}\subset \Q^n$ be a finite point set such
that their set of position vectors $\{p_i\}$ spans $\R^n$. We refer to
$\P$ as a full-dimensional point set in $\R^n$. 

\begin{proposition}
Let $\P=\{p_1,p_2,\ldots,p_m\}\subset \Q^n$ be a full-dimensional
point set. Then there is a unique orthogonal $n\times n$ matrix
$A_\pi$ such that $A_\pi(p_i)=\pi(p_i)$ for each $p_i\in\P$. 
\end{proposition}

\begin{proof}
As $\P$ is full dimensional, we can define a unique matrix $A_\pi$ by
extending $\pi$ linearly to all of $\R^n$. $A_\pi$ can be shown to be
orthogonal as follows. Any vector $x \in \R^n$, $x$ is a linear
combination $\displaystyle\sum_{i=1}^n \sigma_i v_i$ where $v_i \in
\P$.  Then, $\|Ax\|^2 = \displaystyle\sum_{i,j} \sigma_i \sigma_j v_i
A^T A v_j $.  It suffices to observe that $ 2 v_i A^T A v_j =
\|A(v_i-v_j)\| ^2- \|Av_i\|^2 - \|Av_j\|^2 = \|v_i-v_j\|^2 - \|v_i\|^2
- \|v_j\|^2 = 2 v_i^T v_j$ for any vectors $v_i,v_j \in \P$. 
\end{proof}

The geometric automorphism problem is defined below:

\begin{problem}[$\gaut_k$] \leavevmode\newline
\noindent\textbf{Input:} A point set $\{p_1,p_2,\ldots,p_m\}\subset
  \Q^n$ and an orthogonal decomposition of $\R^n=W_1\oplus
  W_2\oplus\cdots\oplus W_r$, where $\dim(W_i)\le k$ and $W_i\perp W_j$
for all $i\ne j$.\\
\noindent\textbf{Parameter:} $k$.\\
\noindent\textbf{Output:} The subgroup $G\le S_m$ consisting of all
automorphisms $\pi$ of the input point set such that the orthogonal
matrix $A_\pi$ stabilizes each subspace $W_i$.
\end{problem}

The $\cmp$ time algorithm for $\evgi_k$ has the following three steps.

\begin{enumerate}
\item We give a polynomial-time reduction from $\evgi_k$ to
  $\gaut_{2k}$.

\item We apply the $\cmp$ time algorithm for $\ggi$ \cite{ar} to give
  a $\cmp$ time reduction from $\gaut_{2k}$ to a special hypergraph
  automorphism problem $\haut$.

\item We give a polynomial-time dynamic programming algorithm for
  $\haut$ by adapting the hypergraph isomorphism algorithm for bounded
  color classes in \cite{adkt}.
\end{enumerate}

\begin{proposition}\label{prop1}
There is a deterministic polynomial-time reduction from $\evgi_k$
with parameter $k$ to $\gaut_{2k}$ with parameter $2k$.
\end{proposition}

\begin{proof}
Let $X=X_1\cup X_2$ be the disjoint union of the input instance
$(X_1,X_2)$ of $\evgi_k$. The adjacency matrix $A_X$ of $X$ is block
diagonal and has the adjacency $A_{X_1}$ and $A_{X_2}$ as its two
blocks along the diagonal. Thus, $A_X$ has the same set of eigenvalues
as $A_{X_1}$ and $A_{X_2}$, and the multiplicity at most
doubles.\footnote{We can assume w.l.o.g.\ that $A_{X_1}$ and $A_{X_2}$
have the same eigenvalues with the same multiplicity as we can check
that in polynomial time.} Clearly, we can decide whether $X_1$ and $X_2$ 
are isomorphic by computing $\Aut(X)$ and checking
if there exists a $\pi \in \Aut(X)$ such that $\pi(X_1)=X_2$ and 
vice-versa.

Furthermore, by Lemma~\ref{prelim1} a permutation $\pi\in\sym(V(X))$
is an automorphism of $X$ if and only if $\pi$ (considered as a linear
map on $\R^{2n}$) preserves each eigenspace of $X$. Let
$\lambda_1,\lambda_2,\ldots,\lambda_r$ be the $r$ eigenvalues of $X$
and $W_1,W_2,\ldots,W_r$ be the corresponding eigenspaces.\footnote{By
  applying suitable numerical methods we can compute each $\lambda_i$
  and basis for each $W_i$ to polynomially many bits of accuracy in
  polynomial time. This suffices for our algorithms.}

Next, we compute the point set $\P=\{p_1,p_2,\ldots,p_{m+2n}\}$
corresponding to the graph $X=(V,E)$, where $|V|=2n$ and $|E|=m$. The
points $p_1,p_2,\ldots,p_{2n}$ are defined by the elementary
$n$-dimensional vectors $e_i\in\R^{2n}, 1\le i\le 2n$. The points
$p_{2n+1},\ldots,p_{2n+m}$ are defined by vectors corresponding to the
edges in $E$ as follows: For each edge $e=\{i,j\}\in E$ the
corresponding point has $1$ in the $i^{th}$ and $j^{th}$ locations and
zeros elsewhere. 

We claim that $\pi \in \Aut(X)$ iff $\pi$ is a geometric automorphism of $\P$. 
Let $\pi$ be any permutation on the vertex set $V(X)$. The action
of the permutation $\pi$ extends (uniquely) to the edge set, and hence
to the point set $\P$ as well. If $\pi\in\Aut(X)$ then, clearly, $\pi$
is a geometric automorphism for the point set $\P$. Conversely, if
$\pi$ is geometric automorphism of the point set $\P$ then it
stabilizes the subset of points $\{p_1,\ldots,p_{2n}\}$ encoding
vertices and the subset $\{p_{2n+1},\ldots,p_{2n+m}\}$ encoding edges
which means $\pi\in\Aut(X)$. This completes the reduction and its
correctness proof.
\end{proof}

\section{The Geometric Automorphism Problem $\gaut_k$}\label{four}

In this section, we introduce some necessary definitions 
and state a useful characterization of a
geometric isomorphism of a set of points.
This will lead to our $\cmp$ time algorithm for $\gaut_k$
which yields the main result for $\evgi_k$ by Proposition~\ref{prop1}.

Let $(\P,W_1,W_2,\ldots,W_r)$ be the instance of
$\gaut_k$. W.l.o.g.\ we can assume that $\P$ is full dimensional.
Otherwise, we can cut down the dimensional of the ambient space $\R^n$
to the dimension of the point set $\P$.

We can assume w.l.o.g.\ that each $W_\ell$ is given by a basis
$u_{\ell 1},u_{\ell 2},\ldots,u_{\ell k_\ell}$ where $k_\ell\le k$ for
all $\ell \in[r]$. 


Each point $p_i\in\P$ has its projection $\proj_\ell(p_i)$ in the
subspace $W_\ell$ defining the projection $\P_\ell=\proj_\ell(\P)$
inside $W_\ell$ of the point set $\P$. For each $p_i\in\P$ we
can uniquely express it as
\[
p_i=\sum_{\ell=1}^r \proj_\ell(p_i).
\] 

Thus we have the projections $\P_1,\P_2,\ldots,\P_r$ of the input
point set $\P$ into the orthogonal subspaces $W_1,W_2,\ldots,W_r$,
respectively. These projections naturally define equivalence
relations on the point set $\P$ as follows.

\begin{definition}
Two points $p_i,p_j\in\P$ are \emph{$(\ell)$-equivalent} if
$\proj_\ell(p_i)=\proj_\ell(p_j)$, and they are
\emph{$[\ell]$-equivalent} if $\proj_t(p_i)=\proj_t(p_j), 1\le t\le
\ell$.
\end{definition}

Since $\R^n=W_1\oplus W_2\oplus\cdots \oplus W_r$ we observe the
following.

\begin{fact}
For any two $p_i,p_j\in\P$ we have $p_i=p_j$ iff $p_i$ and $p_j$ are
$[r]$-equivalent.
\end{fact}

In other words, the common refinement of the $(\ell)$-equivalence
relations, $1\le \ell \le r$, is the identity relation on $\P$, and
the equivalence classes of this refinement are the singleton
sets. Given a permutation $\pi$ on the point set $\P$ we can ask
whether it induces an automorphism on the projection $\P_\ell$ in the
following sense.

A subset $\Delta\subset \P$ of points is an $(\ell)$-equivalence class
of $\P$ if and only if for some point $p\in\P_\ell$ we have
$\Delta=\proj_\ell^{-1}(p)$. Thus, each point in the projected set
$\P_\ell$ represents an $(\ell)$-equivalence class. We say that
permutation $\pi\in\sym(\P)$ \emph{respects} $\P_\ell$ iff for each
$(\ell)$-equivalence class $\Delta\subset \P$ the subset $\pi(\Delta)$
is an $(\ell)$-equivalence class. Suppose $\pi\in\sym(\P)$ is a
permutation that respects $\P_\ell$. Then $\pi$ induces a permutation
$\pi_\ell$ on the point set $\P_\ell$ as follows: for each
$p\in\P_\ell$ its image is
\[
\pi_\ell(p)=\proj_\ell(\pi(\proj_\ell^{-1}(p))).
\]

\begin{definition}
 A permutation $\pi\in\sym(\P)$ is said to be an \emph{induced
   geometric automorphism} on the projection $\P_\ell\subset W_\ell$
 if $\pi$ respects $\P_\ell$ and $\pi_\ell$ is a geometric
 automorphism of the point set $\P_\ell$.
\end{definition}

\begin{lemma}\label{lem-aut1}
Let $(\P,W_1,W_2,\ldots,W_r)$ be an instance of $\gaut_k$ and $\P$ be
full dimensional in $\R^n$. Let $\pi$ be a permutation on $\P$. Then
$\pi$ is a geometric automorphism of $\P$ such that
$A_\pi(W_\ell)=W_\ell$ for each $\ell\in [r]$ if and only if $\pi$ is
an induced automorphism of each $\P_\ell, 1\le \ell \le r$.
\end{lemma}

\begin{proof}
For the forward direction, suppose $\pi$ is a geometric automorphism
of $\P$ such that $A_\pi(W_\ell)=W_\ell$ for each $W_\ell$. We claim
that $\pi$ is an induced automorphism of $\P_\ell$ for each $\ell$.

For any point $p_i\in\P$ we can write 
\[
p_i=\proj_\ell(p_i) + u,
\]
where $u$ is a vector in $W_\ell^\perp$. Since $A_\pi$ stabilizes each
$W_i$, it follows by linearity that
\[
\proj_\ell(A_\pi(p_i))=A_\pi(\proj_\ell(p)).
\]
Hence $A_\pi(\P_\ell)=\P_\ell$ which implies $\pi$ is an induced
automorphism of $\P_\ell$ for each $\ell$.

Conversely, suppose a permutation $\pi$ on $\P$ is an induced
automorphism of each $\P_\ell, 1\le \ell \le r$. Since each $\P_\ell$
is a full-dimensional point set in $W_\ell$, it follows that
$A_\pi(W_\ell)=W_\ell$ for each $\ell$. To see that $\pi$ is a
geometric automorphism of $\P$, let $p_i, p_j\in\P$. We can write
$p_i=\sum_{\ell=1}^r\proj_\ell(p_i)$ and
$p_j=\sum_{\ell=1}^r\proj_\ell(p_j)$. By linearity, we have
$A_\pi(p_i)=\sum_\ell A_\pi(\proj_\ell(p_i))$ and
$A_\pi(p_j)=\sum_\ell A_\pi(\proj_\ell(p_j))$. Hence, by Pythagoras
theorem we have
\begin{eqnarray*}
\|A_\pi(p_i)-A_\pi(p_j)\|^2 & = &
\sum_{\ell=1}^r\|A_\pi(\proj_\ell(p_i))-A_\pi(\proj_\ell(p_j))\|^2\\
& = & \sum_{\ell=1}^r\|\proj_\ell(p_i))-\proj_\ell(p_j))\|^2\\
& = & \|p_i - p_j\|^2,
\end{eqnarray*}
where the third line above follows because $\pi$ is an induced
automorphism of each $\P_\ell$.
\end{proof}


\section{The Hypergraph Automorphism Problem}\label{five}

By Lemma~\ref{lem-aut1} it follows that $\Aut(\P)$ is the group of all
$\pi\in\sym(\P)$ such that $\pi$ is an induced automorphism of each
$\P_\ell, 1\le \ell \le r$. In this section we describe the algorithm
for computing a generating set for $\Aut(\P)$ in $\cmp$ time.

The first step is to reduce $\gaut_k$ in $\cmp$ time to a
hypergraph automorphism problem defined below:

\begin{problem}[$\haut$]\leavevmode\newline
\noindent\textbf{Input:} A hypergraph $X=(V,E)$ and a partition of the
vertex set into color classes $V=V_1\cup V_2\cup \dots \cup V_r$, and
subgroups $G_i\le \sym(V_i), 1\le i\le r$, where each $G_i$ is given
as an explicit list of permutations.\\

\noindent\textbf{Output:} A generating set for $\Aut(X)\cap G_1\times
G_2\times\dots \times G_r$.
\end{problem}

We will give a polynomial-time algorithm for this problem based on a
dynamic programming strategy as used in \cite{adkt}. Before that we
will show that $\gaut_k$ is reducible to $\haut$ in $\cmp$ time.
Combining the two we will obtain the $\cmp$ time algorithm for
$\gaut_k$.

\begin{theorem}\label{geom-gi}
There is a $\cmp$ time reduction from the $\gaut_k$ problem to
$\haut$.
\end{theorem}

\begin{proof}
Let $(\P,W_1,W_2,\ldots,W_r)$ be an instance of $\gaut_k$. In order to
compute $\Aut(\P)$ we first compute each $\P_\ell, \ell\in [r]$. Then,
since $W_\ell$ is $k$-dimensional we can compute the geometric
automorphisms $\Aut(\P_\ell)$ in $\cmp$ time by applying the main
result of \cite{ar}. Indeed, $\Aut(\P_\ell)$ can be explicitly listed
down in $\cmp$ time, also implying that $|\Aut(\P_\ell)|$ is bounded
by $\cmp$. Now, we construct a hypergraph instance $X=(V,E)$ of
$\haut$ as follows: The vertex set $V$ is the disjoint union
$V=\P_1\cup\dots \P_r$, and the explicitly listed groups
$G_\ell=\Aut(\P_\ell), \ell\in [r]$. For each point $p_i\in\P$ we
include a hyperedge $e_p\in E$, where
$e_p=\{\proj_1(p_i),\proj_2(p_i),\ldots,\proj_r(p_i)\}$.  Since the
edges of $X$ encode points in $\P$, the induced action of the
automorphism group $\Aut(X)\cap G_1\times G_2\times\dots\times G_r$ on
the edges of $X$ is in one-to-one correspondence with $\Aut(\P)$ by
Lemma~\ref{lem-aut1}. Hence, we can obtain a generating set for
$\Aut(\P)$.  Clearly, the reduction runs in time $\cmp$. 
\end{proof}

In the polynomial-time algorithm for $\haut$ we will use as
subroutine a polynomial-time algorithm for the following
simple coset intersection problem.

\begin{problem}[Restricted Coset Intersection]
\leavevmode\newline
\noindent\textbf{Input:} Let $V=V_1\uplus V_2\uplus\dots \uplus V_r$
be a partition of the domain into color classes and $G_i\le\sym(V_i)$
be an explicitly listed subgroup of permutations on $V_i$, $1\le i\le
r$. Let $H$ and $H'$ be subgroups of the product group
$G_1\times\dots\times G_r$, where $H$ and $H'$ are given by generating
sets as input. Let $\pi, \pi'\in G_1\times\dots\times G_r$.\\
\noindent\textbf{Output:} The coset intersection $H\pi\cap H'\pi'$
which, if nonempty, is given by a generating set for $H\cap H'$ and
a coset representative $\pi''\in H\pi\cap H'\pi'$.
\end{problem}

\begin{lemma}\label{inter}
The above restricted coset intersection problem has a polynomial-time
algorithm.
\end{lemma}
 
\begin{proof}
We give a sketch of the algorithm which is a simple application of the
classical Schreier-Sims algorithm (mentioned in Section~\ref{two}):
given a permutation group $G\le\sym(\Omega)$ by a generating set and
another permutation $\pi\in\sym(\Omega)$, for any point
$\alpha\in\Omega$ the subcoset of $G\pi$ that fixes the point $\alpha$
can be computed in time polynomial in $|\Omega|$ and the size of the
generating set for $G$. See, e.g.\ \cite{Seress} for details.

In order to compute the intersection $H\pi\cap H'\pi'$, we consider
the product group $H\times H'$ acting on the set
$\Delta=\bigcup_{i=1}^r V_i\times V_i$ component-wise. The permutation
pair $(\pi,\pi')$ too defines a permutation on the set $\Delta$. We
consider now the coset $(H\times H')(\pi,\pi')$ of the group $H\times
H'$. Define the diagonal sets 
\[
D_i=\{(\alpha,\alpha)\mid \alpha\in V_i\}, 1\le i\le r.
\]

The following claim is immediate from the definitions.

\begin{claim}
A pair $(h,h')\in (H\times H')(\pi,\pi')$ maps each $D_i$ to $D_i$ if
and only if $h=h'$ and $h\in H\pi\cap H'\pi'$.
\end{claim}

Thus, in order to compute the coset intersection it suffices to
compute the subcoset 
\[
\{(h,h')\in (H\times H')(\pi,\pi')\mid (h,h')(D_i)=(D_i) 1\le i\le r\}
\]
of the coset $(H\times H')(\pi,\pi')$. Notice that $D_i\subset
V_i\times V_i$ and the elements of the coset $(H\times H')(\pi,\pi')$
restricted to $V_i\times V_i$ are from the group $G_i\times G_i$ which
is polynomially bounded in input size. Let $\Omega$ denote the entire
orbit of $D_i$ under the action of the group $G_i\times G_i$. Clearly,
$|\Omega|\le |G_i|^2$ and therefore is polynomially bounded in input
size and can be computed. Now, $D_i$ is just a point in the set
$\Omega$ and we can compute its pointwise stabilizer subcoset in
$(H\times H')(\pi,\pi')$ by the Schreier-Sims algorithm (as outlined
above) in time polynomial in $|\Omega|$ and the generating sets sizes
of $H$ and $H'$. Repeating this procedure for each $D_i, 1\le i\le r$
yields the subcoset that maps $D_i$ to $D_i$ for each $i$. This
completes the proof sketch. 
\end{proof}

We now describe the polynomial-time algorithm for $\haut$.

\begin{theorem}\label{haut}
There is a polynomial-time algorithm for $\haut$.
\end{theorem}

\begin{proof}
The algorithm is a dynamic programming strategy exactly as in
\cite{adkt}. But, unlike the problem considered in \cite{adkt}, we do
not have bounded-size color classes in our hypergraph
instances. Instead, we have color classes $V_i$ and explicitly listed
subgroups $G_i\le\sym(V_i)$ on each color class and we have to compute
color-class preserving automorphisms $\pi\in\Aut(X)$ that, when
restricted to each color class $V_i$ belong to the corresponding
$G_i$. We now describe the algorithm.

The subproblems of this dynamic programming algorithm involve
hypergraphs $(V,E)$ with multiple hyperedges (i.e., $E$ is a
multi-set). Thus, we may assume that the input $X$ too is a
\emph{multi-hypergraph} given with the vertex set partition
$V=\uplus_{\ell=1}^r V_\ell$, and groups $G_\ell\le\sym(V_\ell)$
explicitly listed as permutations. A bijection $\varphi:V\to V$ is an
automorphism of interest if $\varphi$ maps each $V_\ell$ to $V_\ell$
such that:

\begin{itemize}
\item The permutation $\varphi$ restricted to $V_\ell$ is an
element of the group $G_\ell$.
\item The map induced by $\varphi$ on $E$ preserves the hyperedges
  with their multiplicities (for each hyperedge $e\subseteq V$, $e$
  and $\varphi(e)$ have the same multiplicity in $E$). 
\end{itemize}

We first introduce some notation. For $\ell\in [r]$ and any multi-set
$D$ of hyperedges $e\subseteq V$, let $D_{[\ell]}$ denote the
multi-hypergraph $(V_{[\ell]},\{e\cap V_{[\ell]}\mid e\in D\})$ on
  vertex set $V_{[\ell]}=V_1\uplus\cdots\uplus V_\ell$. Further, let
  $D_\ell$ denote the multi-hypergraph $(V_\ell,\{e\cap V_\ell\mid
    e\in D\})$ on vertex set $V_\ell$. For two multi-hypergraphs
    $D_{[\ell]}$ and $D'_{[\ell]}$ let $\ISO(D_{[\ell]},D'_{[\ell]})$
    denote the coset of all isomorphisms between them that belong to
    $G_1\times\dots\times G_\ell$.

For $\ell\in [r]$ we define an equivalence relation $\equiv_\ell$ on
the hyperedges in $E$: for hyperedges $e_1,e_2\in E$ we say
$e_1\equiv_\ell e_2$ if
\[
e_1\cap V_j= e_2\cap V_j\text{ for } j=\ell+1,\dots,r.
\]

The equivalence classes of $\equiv_\ell$ are called
\emph{$(\ell)$-blocks}. For $\ell \le j$, notice that $\equiv_\ell$ is
a refinement of $\equiv_j$. Thus, if $e_1$ and $e_2$ are in the same
$(\ell)$-block then they are in the same $(j)$-block for all
$j\ge\ell$.

The algorithm works in stages $\ell=0,\dots, r$. In stage $\ell$, the
algorithm considers the multi-hypergraphs $A_{[\ell+1]}$ induced by
the different $(\ell)$-blocks $A$ on the vertex set
$V_{[\ell+1]}$. For each pair of
$(\ell)$-blocks $A,B$ the algorithm computes the cosets
$\ISO(A_{[\ell]},B_{[\ell]})$ (unless $\ell=0$) using the cosets of the form
$\ISO(A^i_{[\ell-1]},B^j_{[\ell-1]})$ computed already.
Finally, for the single $(r)$-block $E$ the algorithm computes the coset 
$\ISO(E_{[r]},E_{[r]})$ which is the desired group $\Aut(X)\cap
G_1\times\dots\times G_r$.

\begin{description}
\item[Stage $0$:] Let $A$ and $B$ be $(0)$-blocks. Then $A$ contains a
  single hyperedge $a$ with multiplicity $|A|$, and $B$ contains $b$
  with multiplicity $|B|$. The coset $\ISO(A_{[1]},B_{[1]})=\emptyset$
  if $\|A\|\neq \|B\|$ or $\|a\cap V_1\|\neq \|b\cap
  V_1\|$. Otherwise, $\ISO(A_{[1]},B_{[1]})\cap G_1$ is a subcoset of
  all elements of $G_1$ that maps $a\cap V_1$ to $b\cap V_1$, which
  can be computed by inspecting the list of elements in $G_1$.
\item[For $\ell:=1$ to $r-1$ do] 
\item[Stages $\ell$:] For each pair $(A,B)$ of $(\ell)$-blocks compute
  the table entry $T(\ell,A,B)=\ISO(A_{[\ell]},B_{[\ell]})$ as follows:
\begin{enumerate}
\item Partition the $(\ell)$-blocks $A$ and $B$ into $(\ell-1)$-blocks
  $A^1,\cdots,A^t$ and $B^1,\cdots,B^{t'}$, respectively. If $t\ne t'$
  then $\ISO(A_{[\ell]},B_{[\ell]})$ is empty. 
\item Otherwise, $t=t'$. Clearly, for all $e \in A^1$, $e \cap V_l$ is
  identical.  Let $a_i= e\cap V_\ell, e\in A^i$ and $b_{i'}=e\cap
  V_\ell, e\in B^{i'}$, for $1\le i, i'\le t$. Let $S_\ell\subset
  G_\ell$ be the subcoset of all permutations $\tau\in G_\ell$ such
  that $\tau$ (injectively) maps the set $\{a_1,a_2,\ldots,a_t\}$ to
    the set $\{b_1,b_2,\ldots,b_t\}$. For each $\tau\in S_|ell$, we
      denote by $\hat{\tau}$ this induced mapping that injectively
      maps the set $\{a_i\mid 1\le i\le t\}$ to
        $\{b_{\hat{\tau}(i)}\mid 1\le i\le t\}$.

We can compute $S_\ell$ in polynomial time since $G_\ell$ is given as
an explicit list as part of the input.

\item For $\tau\in S_\ell$, recall that $A^j_{[\ell-1]}$ and
  $B^{\hat{\tau}(j)}_{[\ell-1]}$ denote the multi-hypergraphs obtained from
  the $(\ell-1)$-blocks $A^j$ and $B^{\hat{\tau}(j)}$, where $j\mapsto
  \hat{\tau}(j)$ for $\tau\in S_\ell$ means that $\tau$ maps $a_j$ to
  $b_{\tau(j)}$. Then it is clear that we have
  \begin{eqnarray}\label{one}
  \ISO(A_{[\ell]},B_{[\ell]})=\bigcup_{\tau\in S_\ell}\bigcap_{j=1}^t
  \ISO(A^j_{[\ell-1]},B^{\hat{\tau}(j)}_{[\ell-1]})\times \{\tau\}
  \end{eqnarray}
where we have already computed the coset
$\ISO(A^j_{[\ell-1]},B^{\pi(j)}_{[\ell-1]})$.
\item In order to compute the coset $\ISO(A_{[\ell]},B_{[\ell]})$ from
  Equation~\ref{one}, we cycle through the polynomially many $\tau\in
  S_\ell$, and compute each coset intersection $\bigcap_{j=1}^t
  \ISO(A^j_{[\ell-1]},B^{\hat{\tau}(j)}_{[\ell-1]})$
  by repeated application of the restricted coset intersection
  algorithm of Lemma~\ref{inter}. We can write a generating set for the 
union of the cosets over all $\tau$ using Fact \ref{union}.
\end{enumerate}
\item[Output:] In the last step, the unique $(r)$-block is the entire set of
  hyperedges $E$, and the table entry
  $T(r,E_{[r]},E_{[r]})=\ISO(E_{[r]},E_{[r]})$.
\end{description}

It is clear from the description that the running time is polynomially
bounded in $|E|, |V|$ and $\max_{1\le \ell\le r}|G_\ell|$. 
\end{proof}

\noindent\textbf{Acknowledgement.}~~ We thank Saket Saurabh for
suggesting the problem of obtaining a faster isomorphism algorithm for
graphs of bounded eigenvalue multiplicity using the $\cmp$ time
algorithm for $\ggi$.

\bibliographystyle{plain}

\begin{thebibliography}{AMW\textsuperscript{+}99}

\bibitem[AMW\textsuperscript{+}88]{A+88}
H. Alt, K. Mehlhorn, H. Wagener, E. Welzl.
\newblock Congruence, similarity, and symmetries of geometric
objects.
\newblock {\em Discrete Computational Geometry,} 3:237-256, 1988.

\bibitem[Ak98]{Aku}
Tatsuya Akutsu.
\newblock On determining the congruence of point sets in d dimensions.
\newblock {\em Computational Geometry}, 9(4):247--256, 1998.

\bibitem[BK00]{BK}
Peter Bra{\ss} and Christian Knauer.
\newblock Testing the congruence of d-dimensional point sets.
\newblock In {\em Symposium on Computational Geometry}, pages 310--314, 2000.

\bibitem[BL83]{BL83}
L{\'a}szl{\'o} Babai and Eugene~M. Luks.
\newblock Canonical labeling of graphs.
\newblock In {\em Proceedings of the ACM STOC Conference,} 
pages 171--183, 1983.


\bibitem[BHZ87]{BHZ}  
Ravi B. Boppana, Johan H{\aa}stad and Stathis Zachos.
\newblock Does co-NP Have Short Interactive Proofs?
\newblock {\em Inf. Process. Lett.}, 25:2, 127-132, 1987.

\bibitem[Luks80]{Luks}
Eugene M. Luks.
\newblock Isomorphism of Graphs of Bounded Valence Can Be Tested in Polynomial Time.
\newblock In {\em Proceedings of the IEEE FOCS Conference}, pages 42-49, 1980.


\bibitem[Mil80]{Miller}
Gary L. Miller.
\newblock Isomorphism Testing for Graphs of Bounded Genus.
\newblock In {\em Proceedings of the ACM STOC Conference,} pages 225-235, 1980.

\bibitem[GM12]{Grohe}
Martin Grohe and D\'aniel Marx.
\newblock Structure theorem and isomorphism test for graphs 
with excluded topological subgraphs.
\newblock {\em 44th ACM Symp. on Theory of Computing,} 
pp. 173-192, 2012.

\bibitem[EP97a]{ep}
S.A. Evdokimov and I.N. Ponomarenko.
\newblock Isomorphism of Coloured Graphs with Slowly Increasing 
Multiplicity of Jordan Blocks.
\newblock {\em Combinatorica} 19(3): 321-333 (1999).

\bibitem[EP97b]{epggi}
S.A. Evdokimov and I.N. Ponomarenko.
\newblock On the geometric graph isomorphism problem.
\newblock {\em Pure and Applied Algebra}, 117-118:253--276, 1997.

\bibitem[BGM82]{bgm}
L{\'a}szl{\'o} Babai,  D. Yu. Grigoryev and David M. Mount.
\newblock Isomorphism of Graphs with Bounded Eigenvalue Multiplicity.
\newblock In {\em Proceedings of the ACM STOC Conference,} pages 310-324, 1982.

\bibitem[Bab86]{bab}
L{\'a}szl{\'o} Babai.
\newblock A Las Vegas-NC Algorithm for isomorphism of graphs with 
bounded multiplicity of eigenvalues.
\newblock In {\em Proceedings of IEEE FOCS Conference,} pages 303-312, 1986.
  
\bibitem[HR14]{HR13} Ishay Haviv and Oded Regev.  
\newblock On the lattice isomorphism problem.  
\newblock In {\em Proceedings of the 25th Annual ACM-SIAM 
Conference,} pages 391-404, SODA 2014.

\bibitem[AR14]{ar}
V. Arvind and Gaurav Rattan.
\newblock The parameterized complexity of geometric graph isomorphism.
\newblock In {\em Proceedings of IPEC Conference, 2014, to appear}.

\bibitem[ADKT10]{adkt}
Vikraman Arvind, Bireswar Das, Johannes K{\"o}bler and Seinosuke Toda.
\newblock Colored Hypergraph Isomorphism is Fixed Parameter Tractable.
\newblock In {\em Proceedings of FSTTCS Conference}, pages 327-337, 2010.

\bibitem[Ser]{Seress}
\'{A}. Seress.
\newblock Permutation Group Algorithms.
\newblock Cambridge University Press, 2003.

\end{thebibliography}

\end{document}